\documentclass[12pt]{article}
\usepackage{setspace}
\usepackage{amsmath}
\usepackage{amssymb}
\usepackage{amsthm}
\usepackage{amsfonts}
\usepackage{natbib}
\usepackage{graphicx}
\numberwithin{equation}{section}
\bibpunct{[}{]}{;}{a}{}{,}
\newtheorem{mydef}{Definition}
\newtheorem{mythe}{Theorem}[section]

\begin{document}

\title{An analogy between the thermal equilibration of a gas mixture and transverse relaxation in magnetic resonance spectroscopy}
\author{\\ \\ \\ Daniel M Packwood\\ \\ Department of Chemistry \\ Graduate School of Science\\ Kyoto University, Kyoto, 606-8502, Japan \\ packwood@kuchem.kyoto-u.ac.jp}
\maketitle
\newpage

\begin{abstract}

We study a gas containing two components, a small component $P$ and a large component $Q$. $P$ is selectively heated to a high temperature and then returns to equilibrium \textit{via} collisions with $Q$. This thermal equilibration process is analysed in a new way. We divide the kinetic energy space of the molecules of $P$ into two regions $F$ and $D$, and show that the molecules of $P$ randomly switch (`oscillate') between the two states as time proceeds due to collisions with the molecules of $Q$. Initially, the molecules of $P$ are all in the state $D$, however because each molecule in $P$ collides with the molecules of $Q$ at different times, the oscillations occur out of step with each other. There is a net destructive interference between the oscillations, and so they are not observed when monitoring the average kinetic energe of the molecules of $P$ as a function of time. We will explain the similarities and differences between this observation and transverse relaxation processes that occur in magnetic resonance spectroscopy. This study employs a stochastic model of elastic collisions between the molecules of $P$ and $Q$, and for completeness we examine its relationship with the two major models of thermal equilibration in statistical physics, namely the Boltzmann equation and the Ornstein-Uhlenbeck process.  


\end{abstract}

\section{Introduction}

Consider a gas with two components $P$ and $Q$, and suppose that $Q$ has contains considerably more molecules than $P$. If $P$ is selectively heated to a high temperature, then collisions between the molecules of $P$ and $Q$ will return $P$ to its equilibrium state. We will call this process thermal equilibration. The starting point for studying thermal equilibration theoretically is the Boltzmann equation. This has been the the subject of considerable research for many years and has produced a vast literature. See [1] and [2] for recent reviews on the topic and a large list of references. Research over the last decade has focused mainly on the rate of approach of the Boltzmann equation to equilibrium, and various bounds have been deduced for spatially homogeneous [3, 4] and inhomogeneous [5, 6, 7] systems. An alternative approach is provided by stochastic models of Brownian motion, usually the Ornstein-Uhlenbeck process and other diffusion processes. Stochastic models of the Brownian motion have also spawned an enormous literature (see, for example, [8] and [9]). The advantage of stochastic models of the Brownian motion over the Boltzmann equation in modelling thermal equilibration is that they can be solved relatively easily. However, this is matched by an important disadvantage, namely that they provide relative little information on the microscopic collision dynamics. Recent studies on this topic have therefore looked at the connection between diffusion processes and more detailed collision models in certain limits. These include Lorentzian gas models [10, 11], heat bath models [12], and models of systems interacting with sequences of classical [13] and quantum systems [14, 15]. Through studies involving the Boltzmann equation and stochastic models, a relatively detailed picture of ther thermal equilibration process is beginning to emerge.

This paper will provide a further insight into the thermal equilibration process. We will show how thermal equilibration can be interpreted in a similar way to an apparently very different non-equilibrium process. Namely, transverse relaxation in a magnetic resonance experiment (also known as $T_2$ relaxation and spin-spin relaxation). Figure 1 illustrates transverse relaxation for a group of spins embedded in a crystal. Omitting several details, the magnetic resonance experiment involves aligining the spins vectors in the $xy$-plane with an electromagnetic pulse at time 0. Following this, the spin vectors rotate in the $xy$-plane together and in near unison. However, random interactions with the environment surrounding each spin causes the angular frequencies of the spins to fluctuate with time. This causes the spin vectors to fall out of alignment with each other. If we plot the $x$ or $y$ component of the spins as a function of time, we therefore see a series of oscillations that gradually fall out of step with each other (Figure 2) [16, 17, 18, 19, 20, 21, 22]. Transverse relaxation can monitored in the laboratory by observing the decay of the net electrical current induced by the rotating spins with time. This decay can be thought of as arising from destructive interference between the oscillations associated with each spin, like those shown in Figure 2. The analogy between thermal equilibration and transverse relaxation is constructed as follows. The kinetic energy space of the molecules in $P$ is divided into a `low energy' region and a `high energy region'. Initially, the molecules in $P$ all start in the same region of the kinetic energy space (Figure 3). However, as time proceeds the molecules of $P$ switch (`oscillate') between the two states due to collisions with the molecules of $Q$ (see Figure 4). Because each molecule in $P$ collides with molecules of $Q$ at different times, the oscillation associated with each molecule falls out of phase with one another. There is a net destructive interference between the oscillations, and therefore the individual oscillations are not seen when the thermal equilibration process is monitored through the average kinetic energy of $P$. The `oscillations' that are described here are not true oscillations because they only involve shifting between two states, rather than a continuous spectrum of states like we have with the $x$ or $y$ component of a spin. However, the terminology is useful for describing the analogy with the transverse relaxation problem. The model that we will use involves randomly occurring elastic collisions between the molecules of $P$ and $Q$, and falls somewhere between the Boltzmann equation and the Ornstein-Uhlenbeck equation in terms of its physical detail. For completeness we will therefore show how the model is related to the Boltzmann and Ornstein-Uhlenbeck equation.

Section 2 describes the model and its basic properties and shows how it is related to the Boltzmann equation and the Ornstein-Uhlenbeck process. Section 3 then describes the analogy between transverse relaxation and thermal equilibration in detail.

\section{Model description}

Consider a one-dimensional, spatially homogeneous gas with two components, a small component $P$ and relatively large component $Q$. The molecules in $P$ and $Q$ have masses $m_p$ and $m_q$, respectively, with $m_p \geq m_q$. Suppose that at time 0 the molecules in $P$ are instantaneously brought out of thermal equilibrium by an outside influence (e.g., a laser pulse). Consider an individual molecule $P_0 \in P$. From time 0, the first particle from $Q$ that $P_0$ collides with is denoted by $Q_1$, the second particle by $Q_2$, and so on. We assume elastic collisions between the molecules of $Q$ and $P$, i.e., upon colliding with $Q_n$ the velocity of $P_0$ becomes

\begin{equation}
V_n = c V_{n-1} + X_n,
\end{equation}

\noindent
where $X_n$ is the velocity of $Q_n$ and

\[
c = \frac{m_p - m_q}{m_p + m_q}.
\]

\noindent
By induction, (2.1) is 

\begin{equation}
V_n = c^n V_0 + \sum_{k=1}^{n} c^{n-k} X_k.
\end{equation}

\noindent
We will assume that $X_1, X_2, \ldots$ are independent $\mathcal{N}(0,\sigma_x^2)$ random variables, and that $V_0 \sim \mathcal{N}(0,\sigma_0^2)$ and is independent of $X_1, X_2, \ldots$. These assumptions are standard assumptions in gas-gas collision models. We will also ignore collisions between molecules in $P$. 

\begin{mydef}
The stochastic process $V^C = \{V_n\}_{n \in \{0, 1, \ldots\}}$, where each $V_n \in V^C$ is given by (2.2), is called the \textit{collision velocity process}. $V^C$ is defined with respect to the probability space $\mathcal{P} = (\Omega, \mathcal{F}, P)$.
\end{mydef}

In the laboratory, equilibration is measured with respect to time, rather than the number of collisions that have occurred. Let $U_1$ be the time of collision between $P_0$ and $Q_1$, $U_2 - U_1$ the length of time between colliding with $Q_1$ and $Q_2$, and so on. We will assume that $U_1, U_2 - U_1 \ldots$ are independent exponential random variables, i.e., for all $k$,

\begin{equation}
P(U_k - U_{k-1} < u) = 1 - e^{-\lambda u}, 
\end{equation}

\noindent
where $\lambda$ is the average frequency of collisions, and that $U_1, U_2 - U_1 \ldots$ are independent of $X_1, X_2, \ldots$ and $V_0$. This means that the number of collisions that $P_0$ has experienced by time $t$ is the value of a Poisson process $N$ at time $t$. This process is also independent of $X_1, X_2, \ldots$ and $V_0$. Thus, the velocity of the particle at time $t$ is simply

\begin{equation}
V(t) = c^{N(t)} V_0 + \sum_{k=1}^{N(t)} c^{N(t)-k} X_k. 
\end{equation}

\begin{mydef}
The stochastic process $V^T = \{V(t)\}_{t \in \mathbb{R}_+}$, where each $V(t) \in V^T$ is given by (2.4), is called the \textit{time velocity process}. $V^T$ is also defined with respect to $\mathcal{P}$. 
\end{mydef}

\noindent
The time parameter of the time-dependent stochastic processes is written inside of the parenthesis to distinguish the time velocity process from the collision velocity process. 

The advantage of defining the `collision velocity process' and the `time velocity process' separately is that relatively difficult calculations on the time velocity process can instead be performed on the collision velocity process. For example, for the collision velocity process we have 

\begin{equation}
E(V_n) = 0,
\end{equation}

\begin{equation}
\mbox{var}(V_n) = c^{2n}\sigma_0^2 + \sigma_x^2 \left(\frac{1 - c^{2n}}{1 - c^2} \right).
\end{equation}

\noindent
for all $n$. To prove the second result, note that $V_n$ is a weighted sum of independent normal random variables, and so

\[
\mbox{var}(V_n) = c^{2n}\sigma_0^2 + \sum_{k=1}^n c^{2(n-k)}\sigma_x^2 = c^{2n}\sigma_0^2 + \sigma_x^2 \sum_{k=0}^{n-1} c^{2k}. 
\]

\noindent
Because $E(E(V(t)^2 \mid N(t))) = E(V_{N(t)}^2)$, we can work out the variance of the time velocity process at time $t$ by conditioning arguments. We find that

\begin{equation}
\mbox{var}(V(t)) = \sigma_0^2 e^{-\lambda t(1 - c^2)} + \frac{\sigma^2}{1 - c^2}\left(1 - e^{-\lambda t (1 - c^2)} \right).
\end{equation}

\noindent
A good definition of an `equilibrium distribution' is also needed to discuss the equilibration process. The following is satisfactory.

\begin{mydef}
The collision velocity process and time velocity process are said to have an \textit{equilibrium distribution} if for all $v \in \mathbb{R}$, $\lim_{n \rightarrow \infty} P(V_n < v)$ and $\lim_{t \rightarrow \infty}P(V(t) < v)$ are well-defined probabilities. 
\end{mydef} 

\begin{mythe}
The collision velocity processes and time velocity processes each have an equilibrium distribution.
\end{mythe}

\begin{proof}
Choose a $v \in \mathbb{R}$. Each $V_n \in V^C$ is a sum of normal random variables (2.2) and are therefore normal random variables with mean and variance given by (2.5) and (2.6). These two quantities determine the distribution. This means that $P(V_n < v) \rightarrow \left(2\pi\sigma^2_\infty\right)^{-1/2} \int_0^v \exp(-u^2/(2\sigma_\infty^2)) du$, which is a well-defined probability.

We can expand the distribution of any $V(t) \in V^T$ with the total rule of probability, i.e.,

\[
P\left(V(t) < v\right) = \sum_{k=0}^{\infty} P\left(V(t) < v \mid N(t) = n\right) P\left(N(t) = n\right).
\]

\noindent
From the Poisson distribution ($P\left(N(t) = n\right) = \exp(-\lambda t) (\lambda t)^n/n!$) we can show that $\lim_{t \rightarrow \infty}P(N(t) = n) = 1$ if $n = \infty$ and is zero otherwise. Therefore,

\begin{eqnarray*}
\lim_{t \rightarrow \infty} P\left(V(t) < v\right) &=& P\left(V(t) < v \mid N(t) = \infty \right) \\
														 &=& \lim_{n \rightarrow \infty} P\left(V_n < v\right),
\end{eqnarray*}

\noindent
which is well-defined by the previous result. The theorem then follows from the fact that $v$ is arbitrary.  
\end{proof}

\noindent
The above proof shows that the equilibrium distributions of both the collision velocity process and the time velocity process are normal with mean zero and variance $\sigma_x^2/(1 - c^2)$.

\subsection{The Boltzmann equation}

To establish the connection of the above model with the Boltzmann equation, let us briefly ignore the results in the previous section and consider how a physicist would approach the problem. We will deliberately gloss over certain mathematical technicalities to stay true to the approach. In its most general form, the Boltzmann equation for a spatially homogeneous gas under no external forces is

\begin{equation}
\frac{\partial f(v,t)}{\partial t} = \left. \frac{\partial f(v,t)}{\partial t}\right|_{coll}
\end{equation}

\noindent
The probability density $f(v,t)$ is proportional to the number of molecules in $P$ with velocities in $(v, v+ \delta v)$ at time $t$, where $\delta v$ is a small constant. The collision term $\partial f(v,t) / \partial t \mid_{coll}$ is equal to the sum of a loss and a gain term. To calculate the loss term, notice that for very small $\delta v$ almost all molecules in $P$ with velocity $v$ will leave $(v,v+\delta v)$ if they experience a collision with a molecule from $Q$. This can be seen directly from equation (1). Therefore, the change in the density of molecules in $(v,v+\delta v)$ due to such collisions during a short time interval of length $\delta t$ is

\begin{equation}
\delta f_{loss} = \left(f(v,t) - p_{coll} (\delta t) f(v,t) \right) - f(v,t),
\end{equation}

\noindent
where $p_{coll}(\delta t)$ is the probability of a collision during the interval $\delta t$. The first term on the right-hand side of (2.9) is the fraction of molecules in $P$ with velocities in $(v, v + \delta v)$ at time $t$ that still have velocities in $(v, v + \delta v)$ after the time period $\delta t$. Expanding $p_{coll}(\delta t)$ to first order in $\delta t$ gives

\begin{equation}
p_{coll}(\delta t) = a \delta t
\end{equation} 

\noindent
where $a$ is a constant. To obtain (2.10), $p_{coll}(0)$ was set to 0. We therefore have $\delta f_{loss} = -a \delta t f(v,t)$. Dividing through by $\delta t$ and taking the limit $\delta v \rightarrow 0$ and $\delta t \rightarrow 0$ gives the loss contribution of $\left.\partial f(v,t)/\partial t \right|_{coll}$ :

\begin{equation}
\left.\frac{\partial f(v,t)}{\partial t}\right|_{loss} = -a f(v,t).
\end{equation}

\noindent
As for the gain term, consider a molecule in $P$ which, at the beginning of the time interval $\delta t$, has a velocity $(v-x)/c$. According to (2.1), this molecule will acquire a velocity in $(v, v+ \delta v)$ if it collides with a molecule from the medium with velocity $x$. Letting $g(u)$ denote the velocity density of the surrounding gas (which is time independent), the change in the velocity probability density of the molecules in $P$ at point $v$ due to these collisions is

\[
\delta f_{gain} = \int_{-\infty}^{\infty} g(x) f((v-x)/c, t) p_{coll}(\delta t) dx.
\]

\noindent
Expanding $p_{coll}(\delta t)$ to first order in $\delta t$, dividing through by $\delta t$ and taking the limit $\delta v \rightarrow 0$ and $\delta t \rightarrow 0$ gives the gain contribution of $\left.\delta f(v,t) / \delta t\right|_{coll}$,

\begin{equation}
\left.\frac{\partial f(v,t)}{\partial t}\right|_{gain} = a \int_{-\infty}^{\infty} g(x) f((v-x)/c,t) dx.
\end{equation}

\noindent
Substituting (2.11) and (2.12) into (2.8) gives

\begin{equation}
\frac{\partial f(v,t)}{\partial t} = a \int_{-\infty}^{\infty} f((v -x)/c,t) g(x) dx - a f(v,t).
\end{equation}

(2.13) is the Boltzmann equation for our elastic collision model. The constant $a$ could be determined by initial conditions, however we will not do this here. Let us now consider the time velocity process described in the previous section. The time velocity process is a stochastic representation of (2.13) in the following sense.

\begin{mythe}
Let $f(v,t)$ denote the probability density function of the random variable $V(t) \in V^T$ at point $v \in \mathbb{R}$. Then $\partial f(v,t) / \partial t$ is given by (2.13) with $a = \lambda$.
\end{mythe}

\begin{proof}
Conditioning $V(t)$ on $\{N(t) = n\}$, we can write

\begin{eqnarray*}
P\left(V(t) < v\right) &=& \sum_{n=0}^{\infty} P\left(V(t) < v\mid N(t) = n\right) P\left(N(t) = n\right) \\
                       &=& \sum_{n=0}^{\infty} P\left(V_n < v\right) P\left(N(t) = n\right).
\end{eqnarray*}

\noindent
Because $V_n$ is a normal random variable, $P\left(V_n < v\right)$ is \textit{a.e} differentiable with respect to $v$. Differentiating the above equation term-by-term then gives a well-defined expression for the probability density of $V(t)$, namely

\[
f(v,t) = \sum_{n=0}^{\infty} f_n(v) P(N(t) = n),
\]

\noindent
where $f_n(v)$ is the probability density of the random variable $V_n$. Differentiating the above with respect to time gives

\begin{equation}
\frac{\partial f(v,t)}{\partial t} = \sum_{n=0}^{\infty} f_n(v) \frac{\partial}{\partial t} P(N(t) = n).
\end{equation}

\noindent
By differentiating the Poisson distribution we can show that $\partial P(N(t) = n) / \partial t = \lambda P(N(t) = n-1) - \lambda P(N(t) = n)$. Therefore,

\begin{equation}
\frac{\partial f(v,t)}{\partial t} = \lambda \sum_{n=0}^{\infty} \left(f_n(v) P(N(t) = n - 1) - f_n(v) P(N(t) =n) \right)
\end{equation}

\noindent
where $P(N(t) = - 1)$ is defined to be zero. The term on the far right of (2.15) is equal to $\lambda f(v,t)$. The first term can be re-written as

\begin{equation}
\sum_{n=0}^{\infty} f_n(v) P(N(t) = n - 1)	= \sum_{n=1}^{\infty} f_{n+1}(v) P(N(t) = n).
\end{equation} 

\noindent
According to (2.1), the density of $V_{n + 1}$ is a convolution of the density of $V_n$ at point $(v - x)/c$ and the density of $X_{n+1}$ at point $x$. That is,

\begin{equation}
f_{n+1}(v) = \int_{-\infty}^{\infty} f_n((v-x)/c) g(x) dx.
\end{equation}

\noindent
Substituting (2.17) into (2.16) and carrying out the sum gives

\begin{equation}
\sum_{n=0}^{\infty} f_n(v) P(N(t) = n-1) = \int_{-\infty}^{\infty} f((v-x)/c) g(x) dx.
\end{equation}

\noindent
Substituting (2.18) into (2.15) gives the result.
\end{proof}

\subsection{The Ornstein-Uhlenbeck process}

Now we will show that as successive collisions between $P_0$ and molecules of $Q$ occur more and more frequently, the time velocity process converge pointwise in probability to the paths of an Ornstein-Uhlenbeck process. The limit that describes this is $\lambda \rightarrow \infty$. A sequence of process $Y_1, Y_2, \ldots$ is said to converge to a process $Z$ pointwise in probability if $P(\mid Y_n(t) - Z(t) \mid > \epsilon) \rightarrow 0$ for all $t$. 

To prove this, we will construct the time velocity process in a slightly different way. Consider the probability space $\mathcal{P} = \left(\Omega, \mathcal{F}, P\right)$. Let $\lambda_1 < \lambda_2 < \cdots$ be a sequence of positive constants such that $\lambda_n \rightarrow \infty$, and define the family $\{\pi_n\}_{n=1}^{\infty}$ such that 

\begin{equation}
\pi_n = \left\{0 < t_1^n < t_2^2 < \cdots \right\},
\end{equation}

\noindent
where $t_k^n \rightarrow \infty$ for each $n$ and 

\[
\mbox{mesh}(\pi_n) = 1/\lambda_n
\]

\noindent
for each $n$. Next, let $W = \{W(t)\}_{t \in \mathbb{R}_+}$ be a Wiener process (following the standard definition, e.g., [25]) on $\mathcal{P}$ and define a sequence $X_1^n, X_2^n, \ldots$ for each $n$ such that

\[
X_k^n = \sigma_x^0 \Delta W(t_{k+1}^n),
\]

\noindent
where $\Delta W(t_{k+1}^n) = W(t_{k+1}^n) - W(t_{k}^n)$ and $\sigma_x^0 > 0$ is a constant. Finally, let the $V^T_1, V^T_2, \ldots$ be a sequence of stochastic processes on $\mathcal{P}$, where for each $V(t)^n \in V^T_n$,

\[
V^n(t) = V_0 c_n^{N(t)^n} + \sum_{k=1}^{N(t)^n} c_n^{N(t)^n - k} X_k^n,
\] 

\noindent
where $c_n$ is a constant and $0 < c_n \leq 1$. As with the time velocity process defined earlier, $V_0 \sim \mathcal{N}(0,\sigma_0^2)$ and each Poisson process $N^n$ is independent of the sequence $X_1^n, X_2^n, \ldots$ and $V_0$. For each $n$, $X_1^n, X_2^n, \ldots$ is a sequence of independent $\mathcal{N}(0, \sigma_x^0/\lambda_n)$, and so for each $n$ the time velocity process defined here satisfies the criteria of Definition 2.

We will now make an important addition to the above construction. Assume that there exists a constant $\alpha < 1$ such that 

\[
c_n = \alpha^{1/\lambda_n}.
\]

\noindent
This means that $c_n \rightarrow 1$ as $\lambda \rightarrow \infty$. This can be interpreted as follows. Recall the definition of $c$ for the processes constructed in section 2,

\[
c = \frac{m_p - m_q}{m_p + m_q}.
\]

\noindent
In this case, $c \rightarrow 1$ as $m_p \rightarrow \infty$. The assumption that $c_n \rightarrow 1$ in the present construction might be therefore be taken to mean that $m_p$ becomes very large compared to $m_q$. In other words, the limiting time velocity process describes a particle with a very large mass, just as is assumed in the usual theories of Brownian motion. We can use this interpretation to understand the definition of $X_k^n$ given above. According to this definition,

\[
\mbox{var} X_k^n \rightarrow 0
\]

This means that the range of velocities in the surrounding gas becomes very narrow in the limit. This can be understood by imagining ourselves riding on $P_0$ as it travels through the medium. Because $P_0$ is very heavy and slow compared to the particles of the surroundings, the particles of the surroundings appear to be moving extremely fast, too fast for us to distinguish their speeds. The assumptions used in this construction can be regarded as a renormalisation of the time velocity process. We have employed similar renormalisations in similar studies of weak convergence to Gaussian processes [22, 24]. 

\begin{mythe}
Let $Y = \{Y(t)\}_{t \in \mathbb{R}_+}$ be the following OU process. For all $Y(t) \in Y$,

\begin{equation}
Y(t) = Y(0) e^{-\theta t} + \eta \int_0^t e^{-\theta (t - s)} dW(s),
\end{equation}

\noindent
where $\theta$ and $\eta$ are positive constants and $W$ is the Wiener process given above. Suppose that $Y(0) = V_0$ \textit{a.s.} and that the integral $\int_0^t \exp(-\theta (t-s)) dW(s)$ is an Ito integral. Then $V_n^T$ converges pointwise in probability to an OU process as $n \rightarrow \infty$.	
\end{mythe}

\begin{proof}
Using the assumption $c_n = \alpha^{1/\lambda_n}$, $V(t)^n$ can be rewritten as

\[
V(t)^n = V_0 \alpha^{N(t)^n/\lambda_n} + \sum_{k=1}^{N(t)^n} \alpha^{(N(t)^n - k)/\lambda_n} \Delta W(t_{k+1}^n).
\]

\noindent
And if we set

\begin{equation}
\theta = -\ln \alpha
\end{equation}

\noindent
we can write (2.20) as

\[
Y(t) = V_0 \alpha^t + \eta \int_{0}^{t} \alpha^{t-s} dW(s).
\]

\noindent
Let

\[
A(t) = \alpha^{N(t)^n/\lambda_n} V_0 - \alpha^t V_0
\]

\noindent
and

\[
B(t) = \sigma_x^0 \sum_{k=1}^{N(t)^n} \alpha^{(N(t)^n - k)/\lambda_n} \Delta W(t_{k+1}^n) - \eta \int_{0}^t \alpha^{t - s}dW(s).
\]

Now, let $U_1^n, U_2^,\ldots$ be the jump times of the Poisson process $N^n$, and $U_t^n = \sup_k(U_k^n \leq t)$. To prove that $A(t)$ and $B(t)$ converge in probability to zero, we will first show that for an arbitrary $n$ and $k \leq N(t)^n$, $\lambda_n U_k^n$ can be brought arbitrary close to $k$ (with respect to an appropriate metric $d(x,y)$ on the space of random variables topologised by convergence in probability, e.g., the Ky Fan metric) with increasing $n$. For an arbitrary $k \leq N(t)^n$ we can use the strong law of large numbers to form the approximation

\[
U_k^n = \sum_{i=1}^k K_i \approx k E(K_1)
\]

\noindent
where the approximation can be made arbitrary accurate (with respect to the metric $d$) by increasing $n$. Thus, $\lambda_n U_k^n$ can be brought arbitrarily close to $k$ by increasing $n$, as claimed.

To prove that $A(t) \rightarrow 0$ in probability, we make the approximation

\[
A(t) \approx V_0 \alpha^{U_t^n} - V_0 \alpha^t.
\]

\noindent
The continuous mapping theorem shows that this approximation can be made as accurate as desired (with respect to $d$) by increasing $n$. Because $\alpha > 1$, we can write

\[
V_0 \alpha^{U_t^n} - V_0 \alpha^t \leq V_0 \alpha^{U_t^n} - V_0 \alpha^{U_{t+1}^n}
\]

\noindent
where $U_{t+1}^n = \inf(U_k^n > U_t^n)$. (2.3) shows that $U_{t+1}^n - U_t^n \rightarrow 0$ in probability, and appealing to the continuous mapping theorem once again shows that $V_0 \alpha^{U_t^n} - V_0 \alpha^t \rightarrow 0$ in probability. We then have that $A(t) \rightarrow 0$ in probability.

As for $\mid B(t)|$, we can write

\begin{eqnarray*}
\mid B(t)\mid &\leq& \left|\sigma_x^0 \sum_{k=1}^{N(t)^n}\alpha^{(N(t)^n - k)/\lambda_n}\Delta W(t_{k+1}^n) - \sigma_x^0 \sum_{k=1}^{N(t)^n}\alpha^{U_t^n - U_k^n} \Delta W(t_{k+1}^n) \right| \mbox{    (T1)} \\
							&+& \left|\sigma_x^0 \sum_{k=1}^{N(t)^n}\alpha^{U_t^n - U_k^n} \Delta W(t_{k+1}^n) - \eta \int_0^t \alpha^{t-s} dW(s) \right| \mbox{    (T2)}
\end{eqnarray*}

\noindent
This is because for any three functions $f(t)$, $g(t)$ and $h(t)$, $\left|f(t) - g(t)\right|$ $=\left|\left(f(t) - h(t)\right) +  \left(h(t) - g(t)\right)\right|$ $\leq \left|f(t) - h(t)\right| + \left|h(t) - g(t)\right|$. Similar to what was done previously, we can make the approximation

\[
\sigma_x^0 \sum_{k=1}^{N(t)^n}\alpha^{(N(t)^n - k)/\lambda_n}\Delta W(t_{k+1}^n) \approx \sigma_x^0 \sum_{k=1}^{N(t)^n}\alpha^{U_t^n - U_k^n} \Delta W(t_{k+1}^n) 
\]

\noindent
which improves in accuracy as $n \rightarrow \infty$. Then $P(\mbox{T1} > \epsilon) \rightarrow 0$ trivially for all $\epsilon > 0$. As for the term (T2), define the family $\{\mu_n\}_{n=1}^\infty$, where $\mu_n = \{0 \leq U_1^n \leq U_2^n \leq \cdots \leq U_t^n \leq t\}$. Each $\mu_n$ is a random partition of $[0,t]$. We will say that the family $\{\mu_n\}_{n=1}^\infty$ \textit{tends to} $[0,t]$ if

\begin{description}
\item[1.] $\lim_{n \rightarrow \infty} U_t^n = t$ in probability.
\item[2.] $\lim_{n \rightarrow \infty} U_1^n = 0$ in probability.
\item[3.] $\sup_{k} \left|U_k^n - U_{k-1}^n\right| \rightarrow 0$ in probability.
\end{description}

\noindent
To check 1, let $\epsilon > 0$ and note that if $U_{t + 1}^n - U_t^n < \epsilon$ then $t - U_t^n < \epsilon$. Because the converse is not necessarily true, this implies that

\[
\left\{t - U_t^n < \epsilon \right\} \supseteq \left\{U_{t + 1}^n - U_t^n < \epsilon \right\},
\]

\noindent
and so by monotonicity,

\[
P\left(t - U_t^n < \epsilon \right) \geq P\left(U_{t+1}^n - U_t^n < \epsilon \right) = 1 - e^{-\lambda_n t} \rightarrow 1.
\]

\noindent
Taking the limit $\epsilon \rightarrow 0$ confirms condition 1. Condition 2 and 3 are true because $P(\mid U_1^n\mid > \epsilon) = e^{-\lambda_n \epsilon} \rightarrow 0$ and $P\left(\sup_{k} \left|U_k^n - U_{k-1}^n\right| > \epsilon \right)$ $= e^{\lambda_n \epsilon} \rightarrow 0$ for all $\epsilon > 0$, according to (2.3). The family $\pi^t = \{\pi_n^t\}_{n=1}^\infty$, where $\pi_n^t = \{0 < t_1^n < t_2^n < \cdots < t_{m+1}^n < t\} \subset \pi_n$, also tends to $[0,t]$. And so for the left-hand term in (T2) we have that 

\begin{eqnarray*}
\sum_{k=1}^{N(t)^n} \alpha^{U_t^n - U_k^n} \Delta W(t_{k+1}^n) &=& \alpha^{U_t^n - t} \sum_{k=1}^{N(t)^n} \alpha^{t - U_k^n} \Delta W(t_{k+1}^n) \\
									&& \rightarrow \lim_{\pi_n^t \rightarrow \pi_{\infty}^t} \sum_{k=1}^n \alpha^{t - t_{k}} \Delta W(t_{k+1}^n)
\end{eqnarray*}

\noindent
in probability. The integral in term (T2) is an Ito integral by assumption, and we can define it as

\[
\int_0^t \alpha^{t - s}dW(s) = \lim_{\pi_n^t \rightarrow \pi_{\infty}^t} \sum_{k=1}^n \alpha^{t - t_{k}} \Delta W(t_{k+1}^n).
\]

\noindent
Therefore, it we set

\begin{equation}
\eta = \sigma_x^0,
\end{equation}

\noindent
then $P(\mbox{T2} > \epsilon) \rightarrow 0$.
\end{proof}

The introduction of this report mentioned several models similar to the above which converge to the Ornstein-Uhlenbeck process and other similar processes. However, Theorem 2.3 has  at least two features that make it interesting. One reason is that pointwise convergence in probability is a particularly strong mode of convergence. Previous research has mainly considered convergence in distribution [10, 11, 12, 13, 14, 15]. The other is that (2.21) and (2.22) in the proof give simple interpretations of the parameters $\theta$ (the friction coefficient) and $\eta$ (the diffusion coefficient). Equation (19) says that 

\[
\theta = -\lambda \ln \left(\frac{m_p}{m_p + m_q} - \frac{m_q}{m_p + m_q}\right),
\]

\noindent
where $\lambda$ and $m_p$ are `very large' parameters. Thus, $\theta$ is related to the average collision frequency and a fractional mass difference. This simply means says that friction on the Brownian particle arises when a relatively heavy particle is ambushed by many lighter particles. This is an alternate interpretation to the standard `heat bath' interpretation of the friction coefficient, in which friction is an consequence of a particle simultaneously interacting with many other particles (see, for example, [26]), rather than a seqeunce of particles like we have here. (2.20) says that the diffusion coefficient is related to the root-mean-square deviation of the velocities in the surroundings, but that it is independent of other quantities such as the frequency of collisions or the mass of the particles.

\section{Oscillations between low and high energy states}

We now want to describe the equilibration process in terms of an `oscillation' between low energy and high energy states. To make this concept more precise, define

\[
F_n = \left\{V_n^2 > V_0^2 \right\}
\]

\noindent
and

\[
D_n = \left\{V_n^2 \leq V_0^2 \right\}.
\]

\noindent
If $F_n$ occurs then $P_0$ has gained a net amount of kinetic energy after undergoing $n$ collisions with the surrounding gas. If $D_n$ occurs, then the test particle has lost a net amount of kinetic energy after undergoing $n$ collisions with the surrounding gas. For convenience, $D_0$ occurs \textit{a.s}. The appropriate analogues for these events in for the time velocity process are

\[
F(t) = \left\{V(t)^2 > V_0^2 \right\}
\]

\noindent
and

\[
D(t) = \left\{V(t)^2 \leq V_0^2 \right\}.
\]

\noindent
Now, let

\[
F = \left\{v \in \mathbb{R} : V_0^2 - v^2 < 0 \right\},
\]

\noindent
and

\[
D = \left\{v \in \mathbb{R} : V_0^2 - v^2 \geq 0 \right\}.
\]

\noindent
$D$ and $F$ are the `low energy' and `high energy' states that we are interested in. If $D_n$ occurs then $V_n \in D$ and if $F_n$ occurs then $V_n \in F$. Similarly, if $D(t)$ occurs then $V(t) \in D$ and if $F(t)$ occurs then $V(t) \in F$. Clearly, $D(0)$ occurs \textit{a.s}.

\subsection{Recurrence of $D$ and $F$}

If $D$ and $F$ are recurrent, then $P_0$ enters and exits the states $F$ and $D$ infinitely often (\textit{i.o.}) with probability 1 as it travels through $Q$. In this situation, $P_0$ therefore switches between high and low energy states over time. We will refer to this switching as an `oscillation'. While this is not a true oscillation, the terminology is useful in order to create the analogy with transverse relaxation in magnetic resonance spectroscopy. The occurrence of such oscillations is proven in this section. Recall that for a sequence of events $E_1, E_2, \ldots$,

\begin{eqnarray}
\{E_n \mbox{ } i.o.\} &=& \bigcap_{n=1}^{\infty} \bigcup_{m=n}^{\infty} E_n \nonumber \\
                      &=& \left\{\omega \in \Omega \mbox{ that belong to infinitely many of the } E_n \right\}. 
\end{eqnarray}

\begin{mythe}
$P(D_n \mbox{ } i.o) = P(F_n \mbox{ } i.o) = 1$
\end{mythe}

\begin{proof}
The event $\{D_n \mbox{ } i.o\}$ is permutable because its occurrence will not be affected by finite permutations of the indices of $X_1, X_2, \ldots$ [25]. Suppose that $P(D_n \mbox{ } i.o) < 1$. Then the Hewitt-Savage 0-1 law implies that $P(D_n \mbox{ } i.o) = 0$. There therefore exists an $m < \infty$ such that $P(D_n) = 0$ for all $n > m$. In other words,

\begin{equation}
V_0^2 - V_n^2 < 0
\end{equation} 

\noindent
for all $n > m$ with probability 1. (2.3) shows that 

\[
c^n V_0 +  \sum_{k=1}^n X_k = c^{n-m}V_m + Y_{n-m}
\]

\noindent
where

\[
Y_{n-m} =  \sum_{k = m+1}^n c^{n-k} X_k.
\]

\noindent
Substituting this into (3.2) gives

\[
2c^{n-m} V_m Y_{n-m} - Y_{n-m} > V_0^2 - c^{2(n-m)}V_m^2.
\]

\noindent
Taking the expected value and noting that $V_m$ and $Y_{n-m}$ are independent mean zero random variables, we obtain

\[
0 > \sigma_0^2 - c^{2(n-m)}E(V_m^2).
\]

\noindent
In the limit $(n - m) \rightarrow \infty$, the above becomes $0 > \sigma_0^2$, which is nonsensical. So we conclude that $P(D_n \mbox{ } i.o) = 1$. A similar argument gives $P(F_n \mbox{ } i.o) = 1$.
\end{proof}

Now we need to establish the connection between this result and $V^T$ \textit{entering} $F$ and $D$ \textit{i.o}. A very reasonable definition of this event is

\[
\left\{V^T \mbox{ enters } D \mbox{ } i.o. \right\} = \left\{V_{n} \in D, V_{n - 1} \in F \mbox{ } i.o. \right\},
\]

\noindent
and similarly for $\left\{V^T \mbox { enters } F \mbox{ } i.o.\right\}$. However, Theorem 3.2 immediately implies that $\left\{V_{n} \in D, V_{n - 1} \in F \mbox{ } i.o. \right\}$ occurs with probability 1. And so we have the following theorem.

\begin{mythe}
$P\left(V \mbox{ enters } F \mbox{ } i.o.\right) = P\left(V \mbox{ enters } D \mbox{ } i.o.\right) = 1$
\end{mythe}  

While all molecules in $P$ start in the $D$ state, they each collide with molecules from $Q$ at different times, and therefore the oscillations between $F$ and $D$ eventually fall out of phase with one another (Figure 4). There is a net destructive interference between the oscillations, and therefore they are not seen if we study the average kinetic energy of the molecules of $P$ directly. In this sense, the thermal equilibration of a gas and transverse relaxation of spins in magnetic resonance spectroscopy can be understood in a similar way. 

\subsection{Frequency of oscillations between $F$ and $D$}

The next two sections will characterise the oscillations by studying their frequency and period. We will restrict attention to the expected value of these quantities. Define the \emph{collision crossing number},

\[
C_n = \left|\left\{\mbox{$1 \leq m \leq n$ such that ($V_{m-1} \in D$, $V_{m} \in F$) or ($V_{m-1} \in F$, $V_{m} \in D)$}\right\}\right|.
\]

\noindent
$C_n$ is the number of collisions out of $n$ collisions that cause $P_0$ to cross the boundary between $F$ or $D$. The random variable

\[
W_n = C_n/n
\]

\noindent
measures the frequency at which $P_0$ enters and exits $F$ as a function of the number of collisions. $W_n$ is called the \emph{collision crossing frequency}.

\begin{mythe}
$E(W_1), E(W_2), \ldots$ is a strictly and monotoncially increasing sequence and $E(W_n) \rightarrow \alpha$, where $\alpha \leq 1$ is a positive constant.
\end{mythe}

\begin{proof}
The second part of the theorem follows from the first \textit{via} the completeness axiom and the fact that $W_n \leq 1$ \textit{a.s}. To prove the first part, define the family $\mathcal{A}_1, \mathcal{A}_2, \ldots$, where

\[
\mathcal{A}_n = \{0, 1/n, 2/n, \ldots, 1\}.
\] 

\noindent
$\mathcal{A}_n$ is the state space of $W_n$. Let $a_i^n$ denote be $i$th element of $\mathcal{A}_n$ (i.e., $a_k^n = (k-1)/n$ for $k = 1, 2, \ldots, n+1$) and $b_i^n = n a_i^n$. Now, fix an arbitrary $n$ and arbitrary $a_i^n \in \mathcal{A}_n$. Let

\[
H_{n+1} = \left\{V_{n+1} \in F \cap V_{n} \in D\right\} \cup \left\{V_{n+1} \in D \cap V_{n} \in F\right\}. 
\]

\noindent
If $H_{n+1}$ occurs then the collision velocity process crosses the boundary between $F$ and $D$ at the $(n+1)$th collision. Using this event, the event $\{C_{n+1} > b_i^n\}$ can be decomposed as follows:

\[
\left\{C_{n+1} > b_i^n\right\} = \left\{C_n > b_i^n\right\} \cup \left(\left\{C_n = b_i^n\right\} \cap H_{n+1}\right).
\]
 
\noindent
Because $\{C_{n} > b_i^n\}$ and $(\{C_n = b_i^n\} \cap H_{n+1})$ are mutually exclusive, we can then write

\begin{equation}
P\left(C_{n+1} > b_i^n\right) = P\left(C_n > b_i^n\right) + P\left(\left\{C_n = b_i^n\right\} \cap H_{n+1}\right).
\end{equation}

\noindent
To form a strict inequality from this equation, we will show that $P(\{C_n = b_i^n\} \cap H_{n+1}) > 0$. Without loss of generality, suppose that $b_i^n$ is odd. Because $V_0 \in D$ \textit{a.s}, this means that if $\{C_n = b_i^n\}$ then $V_n \in F$ with probability 1. And so 

\[
P\left(\left\{C_n = b_i^n\right\} \cap H_{n+1}\right) = P\left(\left\{C_n = b_i^n\right\} \cap \left\{V_{n+1} \in D \right\} \right).
\]

\noindent
Let us suppose that $P(\{C_n = b_i^n\} \cap \{V_{n+1} \in D \}) = 0$. According to de Morgan's law, this means that

\[
P\left(\left\{C_n = b_i^n\right\}^c \cup \left\{V_{n+1} \in F \right\} \right) = 1,
\]

\noindent
where $A^c$ is the compliment of event $A$. For any two events $A_1, A_2 \in \Omega$, we have the elementary relationship $P(A_1 \cup A_2) = P(A_1) + P(A_2) - P(A_2 \cap A_1)$. If we set $P(A_1) = \epsilon_1$ and $P(A_2 \cap A_1) = \epsilon_2$, then this relationship can be written as $P(A_2) = P(A_1 \cup A_2) - (\epsilon_1 - \epsilon_2)$. Applying this to the above, with $A_1 = \{C_n = b_i^n\}^c$ and $A_2 = \{V_{n+1} \in D \}$, we have that 

\[
P\left(V_{n+1} \in D \right) = 1 - \left(\epsilon_1 - \epsilon_2\right).
\]

\noindent
According to this and (2.1), with probability $1 - (\epsilon_1 - \epsilon_2)$ the inequality

\[
V_0^2 - \left(cV_n + X_{n+1}\right)^2 \leq 0.
\]

\noindent
is satisfied. However, $X_{n+1}$ is independent of $V_n$ and $V_0$ and so this inequality will not hold with probability $1 - (\epsilon_1 - \epsilon_2)$ for general $\epsilon_1$ and $\epsilon_2$. So we conclude that $P(\{C_n = b_i^n\} \cap H_{n+1}) > 0$. From (3.3) we then obtain

\[
P\left(C_n > b_i^n\right) < P\left(C_{n+1} > b_i^n\right).
\]

\noindent
In other words,

\[
P\left(W_{n} > a_i^n\right) < P\left(W_{n+1} > a_i^n n/(n+1) \right).
\]

\noindent
To complete the proof, note that $a_{i-1}^{n+1} < a_{i}^n n/(n+1)$. This means that $P(W_{n+1} > a_i^n n/(n+1)) \leq P(W_{n+1} > a_{i-1}^{n+1})$. Thus, summing both sides of the above over $\mathcal{A}_n$ and using the fact that the cardinality of $\mathcal{A}_{n+1}$ is greater than the cardinality of $\mathcal{A}_{n}$, we find that

\begin{eqnarray*}
E(W_n) &<& \sum_{a_i^n \in \mathcal{A}_n} P\left(W_{n+1} > a_i^n n/(n+1)\right) \\
			 &\leq& \sum_{a_{i-1}^{n+1} \in \mathcal{A}_{n+1}} P\left(W_{n+1} > a_{i-1}^{n+1} \right).
\end{eqnarray*}

\noindent
The sum on the right hand side goes over all elements of $\mathcal{A}_{n+1}$ expect for 1. However, $W_{n+1} \leq 1$ \textit{a.s}. Therefore,

\[
\sum_{a_{i-1}^{n+1} \in \mathcal{A}_{n+1}} P\left(W_{n+1} > a_{i-1}^{n+1} \right) = \sum_{a_i^{n+1} \in \mathcal{A}_{n+1}} P\left(W_{n+1} > a_i^{n+1}\right), 
\]

\noindent
and hence we obtain $E(W_n) < E(W_{n+1})$. Applying this result to the cases $n=1, 2, \ldots$ in order, we can complete the proof by induction.
\end{proof}

This result can also be extended to the time velocity process in a straightforward way. Define the \emph{time crossing frequency} at time $t$,

\[
W(t) = C_{N(t)}/N(t)
\]

\noindent
$W(t)$ is the fraction of collisions that have caused the particle to enter or exit the state $F$ (or $D$) out of the total number of collisions that have occurred by time $t$. 

\begin{mythe}
Let $t_1, t_2, \ldots$ be any sequence of times such that $0 \leq t_1 < t_2 < \cdots$. Then $E(W(t_1)), E(W(t_2)), \ldots$ is a strictly monotonically increasing sequence and $E(W(t_n)) \rightarrow \alpha$, where $\alpha \leq 1$ is a positive constant.
\end{mythe}

\begin{proof}
Under the condition $\{N(t) = n\}$, $W(t) = W_n$ \textit{a.s.}. According to Theorem 3.3 and the fact that $N(s) \leq N(t)$ for all $s < t$ \textit{a.s.},

\[
E(W(t) \mid N(t)) \geq E(W(s) \mid N(s))
\]

\noindent
\textit{a.s}. Now, if $f$ is a strictly and monotonically increasing function, then $E(f(N(t))) > E(f(N(s)))$. Applying this to the above result gives $E(E(W(t) \mid N(t))) > E(E(W(s) \mid N(s)))$. Equivalently, we have that $E(W(t)) > E(W(s))$ for all $s < t$. 

To prove the second part, note that for every $t$ $E(W(t) \mid N(t)) \geq 0$ and that $E(W(t) \mid N(t)) \rightarrow \alpha \leq 1$ according to Theorem 3.3. Using the monotone convergence theorem, we therefore have that $E(E(W(t) \mid N(t))) \rightarrow \alpha$, or that $E(W(t)) \rightarrow \alpha$.
\end{proof}

Theorem 3.5 says that on average the molecules of $P$ oscillate more and more quickly between the high- and low kinetic energy states as $P$ approaches equilibrium. We can understand this with the following analogy. Imagine that we have a hot piece of metal that we want to cool to a particular temperature. So we dip it into a bucket of cold water for a length of time $T_1$ to cool it down. However, then we find that the metal is now too cold, so we heat it up with a flame for a shorter time period $T_2$ to warm it up. However, now the metal is too hot, so then we put it back into the bucket of water, and so on. Each step adjusts the temperature of the metal, but always ends up undershooting or overshooting the target temperature. As we get closer to the target temperature, the length of time that we need to heat or cool the metal for becomes shorter and shorter. The oscillations of the molecules in $P$ between high and low energy states can be understood in the same way. Initially the molecule has too much kinetic energy, so the surroundings work to remove energy from it by collisions. Eventually, the molecule has too little kinetic energy, and so the surroundings work to provide it with more kinetic energy. However, the molecule then ends up acquiring too much kinetic energy, and so on. The analogy between thermal equilibration and transverse relaxation ends with Theorem 3.5, because in the case of transverse relaxation the average oscillation frequency is time independent [16].

\subsection{Estimate of the first hitting time to state $F$}

Computing the average period of the oscillations appears does not appear to be possible in general, particularly because the average periods become shorter with each successive oscillation (Theorem 3.5). We will instead estimate the average length of the first period and use this as a `reference value' to gauge the magitude to the other periods. The average length of the first period is the expected value of the first hitting time to state $F$. Let

\[
\tau_1 = \inf \left(t \in \mathbb{R}_+ : V(t) \in F \right).
\]

\noindent
and

\[
N_1 = \min\left(n : V_n \in F \right).
\]

\noindent
Putting these together, we have

\[
\tau_1 = \sum_{i=1}^{N_1} K_i,
\]

\noindent
where $K_1, \ldots, K_{N_1}$ are a sequence of independent exponential random variables with expectation $E(K) = 1/\lambda$.

\begin{mythe}
For $N_1$ and $\tau_1$ defined above, $E(N_1) < \infty$ and $E(\tau_1) = E(N_1)/\lambda$.
\end{mythe}

\begin{proof}
The second part of the theorem follows from the first part and Wald's equation [25]. To prove the first part, it is sufficient to show that

\[
\lim_{m \rightarrow \infty} P\left(N_1 > m\right) = 0.
\]

\noindent
Now, because $D_0$ occurs \textit{a.s.}, 

\[
\left\{N_1 > m\right\} = \bigcap_{k=1}^{m} D_k = \bigcap_{k=1}^{m} F_k^c,
\]

\noindent
and so, by de Morgan's law, 

\[
\left\{N_1 > m\right\} = \left\{\bigcup_{k=1}^m F_k\right\}^c.
\]

\noindent
Because $\bigcup_{k=1}^m F_k \supseteq \bigcap_{n=1}^m \bigcup_{k=n}^m F_k$, we can form the inequality

\[
\left\{N_1 > m\right\} \subseteq \left\{\bigcap_{n=1}^m \bigcup_{k=n}^m F_k\right\}^c,
\]

\noindent
And so, by monotonicity 

\[
P\left(N_1 > m\right) \leq P\left(\left\{\bigcap_{n=1}^m \bigcup_{k=n}^m F_k\right\}^c\right).
\]

\noindent
Theorem 3.2 then implies that

\[
\lim_{m \rightarrow \infty} P\left(N_1 > m\right) \leq P\left(\left\{F_n \mbox{ } i.o\right\}^c\right) = 0.
\]
\end{proof}

\begin{mythe}
For $\tau_1$ defined as above,

\begin{equation}
E(\tau_1) \geq \frac{\sigma_0^2}{\lambda\sigma_x^2}\left(1 - \left(\frac{m_p - m_q}{m_p + m_q}\right)^2\right)
\end{equation}

\end{mythe}

\begin{proof}
Let 

\[
S_{N_1} = \sum_{k=1}^{N_1} c^{N_1 - k} X_k,
\]

\noindent
so that we can re-write $V_{N_1}$ as 

\begin{equation}
V_{N_1} = c^{N_1} V_0 +  S_{N_1}.
\end{equation}

\noindent
By the definition of $N_1$, $V_{N_1}^2 > V_0^2$ \textit{a.s}.  Alternatively, we can write

\[
c^{2N_1} V_0^2 + ^2 S_{N_1}^2 + 2c^{N_1} S_{N_1} > V_0^2.
\] 

\noindent
Using the fact that $c^{2{N_1}} \leq c^{2}$ \textit{a.s.} and that $N_1$ and $V_0$ are independent, taking the expected value of the above gives

\begin{equation}
c^2\sigma_0^2 + ^2 E(S_{N_1}^2) + 2E(c^{N_1}S_{N_1}) > \sigma_0^2.
\end{equation}

\noindent
We need to deal with the expected values on the left hand side. First consider $E(c^{N_1} S_{N_1})$. For almost all $\omega \in \Omega$, either $\omega \in \{V_{N_1} = V_0\}$ or $\omega \in \{V_{N_1} = -V_0\}$. These events are mutually exclusive, and so we can write

\begin{eqnarray*}
E(S_{N_1}) &=& E(S_{N_1} \mid V_{N_1} - V_0 = -2V_0) P(V_{N_1} - V_0 = -2V_0) \\
           &	+& E(S_{N_1} \mid V_{N_1} - V_0 = 0) P(V_{N_1} - V_0 = 0). 
\end{eqnarray*}

\noindent
Now, from (3.5) we can form the random variable

\[
V_{N_1} - V_0 =  S_{N_1} - (1 - c^{N_1})V_0
\]

\noindent
Taking the expected value of this under the condition $\{V_{N_1} - V_0 = -2V_0\}$ gives (because $E(V_0) = 0$),

\[
E(S_{N_1} \mid V_{N_1} - V_0 = -2V_0) = 0.
\]

\noindent
Similarly, evaluating the expected value under the condition $\{V_{N_1} - V_0 = -2V_0\}$ gives $E(S_{N_1} \mid V_{N_1} - V_0 = 0) = 0$. So we have that $E(S_{N_1}) = 0$. Therefore, $E(c^{N_1}S_{N_1}) = E(E(c^{N_1}S_{N_1} \mid N_1)) = E(c^{N_1} E(S_{N_1} \mid N_1)) = 0$.

Now we will prove that $E(S_{N_1}^2) \leq E(N_1) \sigma^2$. Following the method on pg. 187 of reference [25], we can write

\[
S_{N_1 \wedge n}^2 = c^2 S_{N_1 \wedge (n-1)}^2 + \left(2c X_n S_{n-1} + X_n^2 \right) \textbf{1}(N_1 \geq n),
\]

\noindent
Because $N_1$ is a stopping time, $\textbf{1}(N_1 \geq n) = \textbf{1}(N_1 > n-1)$ depends at most upon $X_1, \ldots, X_{n-1}$. Moreover, because $S_{n-1}$ and $X_n$ are independent,

\[
E\left(S_{N_1 \wedge n}^2\right) = c^2 E\left(S_{N_1 \wedge (n-1)}^2\right) + \sigma^2 P(N_1 \geq n).
\] 

\noindent
By induction,

\begin{eqnarray*}
E\left(S_{N_1 \wedge n}^2\right) &=& \sigma_x^2 P(N_1 \geq n) + c^2 \sigma_x^2 P(N_1 \geq n-1) + \ldots + c^{2n} \sigma_x^{2} P(N \geq 0) \\
																 &\leq& \sigma_x^2 \sum_{k=1}^n P(N_1 \geq k). 
\end{eqnarray*}

\noindent
Taking the limit leads to $E(S_{N_1}^2) \leq \sigma_x^2 E(N_1)$.

Substituting these results into (3.6) gives and using the second part of Theorem 3.6 completes the proof.
\end{proof}

The above proof shows that the tightness of (3.4) is measured by how close $c$ is to 1. As discussed in section 2.3, this condition is satisfied when the particles of $P$ are very heavy, such as occurs in the Brownian motion. In fact, taking the limit $c \rightarrow 1$ gives 

\[
E(\tau_1) = 0,
\]

\noindent
in other words, a Brownian particle will almost immediately shift into the $F$ state upon being displaced from equilibrium. According to Theorem 3.4, the Brownian particle will oscillate extremely rapidly between the $F$ and $D$ states. Across the ensemble $P$ of Brownian particles, these oscillations will occur slightly out of phase with one another, and so the average kinetic energy of the particles in $P$ will quickly decay to zero. 

Let us now consider the other extreme where $m_p = m_q$. The estimate in (3.6) simplifies to $E(\tau_1) \geq \sigma_0^2/\lambda \sigma_x^2$. Moreover, elementary kinetic theory tells us that the $\sigma_0^2/\sigma_x^2 = T_p/T_q$, where $T_p$ is the initial temperature of $P$ and $T_q$ the temperature of $Q$, respectively. And so we have the inequality

\[
E(\tau_1) \geq \frac{T_p}{\lambda T_q}.
\]

\noindent
The parameter $\lambda$ will have a dependence upon $T_p$ and $T_q$ as well, and so the relationship between frequency and temperature is not as simple as the above suggests. If $T_p$ is much greater than $T_q$, then the molecules of $P$ will initially be traveling much faster than those of $Q$ and so the molecules of $Q$ will appear stationary from the point-of-view of a molecule in $P$. Elementary kinetic theory then says that $\lambda \propto 1/T_p$, and so

\[
E(\tau_1) \geq a\frac{T_p^2}{T_q}
\]

\noindent
where $a > 0$ is a constant independent of $T_p$ and $T_q$. This lower bound is very large, and so displacing $P$ far from equilibrium will generate oscillations with a small oscillation frequency. The experimental prediction of this work is that if $P$ contains very few molecules (so that destructive interference between the individual oscillations is minimised), then the individual oscillations between $F$ and $D$ may be in phase long enough to have a noticible affect on the data. In particular, we would expect that thermal equilibration would take a longer length of time with a smaller ensemble size for $P$. This is analogous to how ensembles of a very small number of spins stay in phase with each other over a relatively long time period during transverse relaxation and phase decoherence measurements [27, 28]. The effect of ensemble size on the rate of thermal equilibration might be an interesting topic for a future study.

\section{Final remarks}

Through a series of mathematical arguments we have shown that thermal equilibration of a high temperature component of a two-component gas can be interpreted in a similar way to transverse relaxation in magnetic resonance spectroscopy. We did this by dividing the kinetic energy state space of the molecules of the high temperature component into a `high energy' and a `low energy' component, and by then showing that the molecules shift (`oscillate') between the two components forever over time. Because each molecule in the gas enters and exits these states at a different time, the `oscillations' occur out of phase and interfere destructively with one another, and so are not seen when monitoring the equilibration process \textit{via} the average kinetic energy of the gas in the laboratory. We also suggested that the oscillations may stay in-phase long enough to have an observable effect on the thermal equilibration process if the high temperature component contained very few molecules. We also characterised the oscillation by studying its frequency, and showed that the frequency strictly increases with time. This is where transverse relaxation and thermal equilibration differ, because in the former case the average oscillation frequencies are usually constant in time [16]. The stochastic model that we used is not standard in the statistical physics literature, however we demonstrated that it is consistent with both the two key models of thermal equilibration (Boltzmann equation and Brownian motion) and so there is no reason to doubt that the model is physically meaningful.

The transverse relaxation-like `oscillatory' interpretation of the thermal equilibration presented here offers a new and intuitive way of looking at the thermal equilibration process. It may be a useful starting point for creating models of the thermal equilibration process, and also suggests that there might be other connections between thermal equilibration and other seemingly disparate non-equilibrium phenomena in statistical physics. \\

\noindent
\textbf{Acknowledgements}
DMP is supported by a Japan Society for the Promotion of Science Postdoctoral Fellowship. Prof. Yoshitaka Tanimura is thanked for useful comments.\\

\noindent
\textbf{References} \\

\noindent
[1] Perthame B 2004 Mathematical tools for kinetic equations \textit{Bull. Amer. Math. Soc.} \textbf{41} 205 - 244.\\

\noindent
[2] Alexandre R 2009 A reivew of the Boltzmann equation with singular kernels \textit{Kinetic and Related Models} \textbf{2} 551 - 645 \\

\noindent
[3] Mouhot C 2007 Quantitative linearized study of the Boltzmann collision operator and applications \textit{Comm. Math. Sci.} \textbf{5} 73-86\\

\noindent
[4] Bassetti F, Ladelli L and Toscani G 2011 Kinetic models with randomly perturbed binary collisions \textit{J. Stat. Phys.} \textbf{142} 686 - 709\\

\noindent
[5] Desvillettes L and Villani C 2002 On a variant of Korn's inequality arising in statistical mechanics \textit{ESAIM: Cont. Optim. Calc. Var.} \textbf{8} 603 - 619\\

\noindent
[6] Desvillettes L and Villani C 2005 On the trend to global equilibrium for spatially inhomogeneous kinetic systems: The Boltzmann equation \textit{Invent. Math.} \textbf{159} 245 - 316\\

\noindent
[7] Carlen E A, Carvalho M C and Lu X 2009 On the strong convergence to equilibrium for the Boltzmann equation with soft potentials \textit{J. Stat. Phys.} \textbf{135}, 681 - 736\\

\noindent
[8] Gardiner C 2009 \textit{Stochastic Methods: A Handbook for the Natural and Social Sciences} (Berlin: Springer)\\

\noindent
[9] Oksendal B 2007 \textit{Stochastic Differential Equations: An Introduction with Applications} (Berlin: Springer)\\

\noindent
[10] Dolgopyat D and Koralov L 2009 Motion in a random force field \textit{Nonlinearity} \textbf{22} 187 - 211\\

\noindent
[11] Aguer D, De Bievre S, Lafitte P and Parris P E 2010 Classical motion in force fields with short range correlations \textit{J. Stat. Phys.} \textbf{138} 780-814\\

\noindent
[12] Kupferman R 2004 Fractional kinetics in Kac-Zwanzig heat bath models \textit{J. Stat. Phys.} \textbf{114} 291 - 326\\

\noindent
[13] Packwood D M 2010 The Ornstein-Uhlenbeck equation as a limiting case of a successive interactions model \textit{J. Phys. A: Math. Theor.} \textbf{43} 465001 - 465014\\

\noindent
[14] Attal S and Pautrat Y 2006 From repeated to continuous quantum interactions \textit{Ann. Henri Pioncare} \textbf{7} 59 - 104\\

\noindent
[15] Pellegrini C and Petruccione F 2009 Non-Markovian quantum repeated interactions and measurements \textit{J. Phys. A: Math. Theor.} \textbf{42} 425304 - 425325\\

\noindent
[16] Kubo R 1969 Stochastic theory of line shape \textit{Adv. Chem. Phys.} \textbf{15} 101\\

\noindent
[17] Breuer H-P and Petruccione F 2002 \textit{The Theory of Open Quantum Systems} (New York: Oxford University Press)\\

\noindent
[18] Tyryshkin A M, Morton J J L, Benjamin S C, Ardavan A, Briggs G A D, Ager J W and Lyon S A 2006 Coherence of spin qubits in silicon \textit{J. Phys.: Condens. Matter} \textbf{18} S783 - S794\\

\noindent
[19] Endo T 1987 Quantum theory of phase relaxation \textit{J. Phys. Soc. Japan} \textbf{56} 1684 - 1692 \\

\noindent
[20] Ban M, Shibata F and Kitajima S 2007 On phase relaxation processes \textit{J. Mod. Opt.} \textbf{54} 555 - 564\\

\noindent
[21] Kitajima S, Ban M and Shibata F 2010 Theory of decoherence control in a fluctuating environment \textit{J. Phys. B: At. Mol. Opt. Phys.} \textbf{43} 135504 - 135518\\

\noindent
[22] Packwood D M and Tanimura Y 2011 Non-Gaussian stochastic dynamics of spins and oscillators: A continuous-time random walk approach \textit{Phys. Rev. E} \textbf{84} 61111 - 61124\\

\noindent
[23] Packwood D M and Tanimura Y 2012 Phase decoherence by a continuous-time random walk process \textit{In preparation}. Preprint: arXiv:1205.0296v1 \\

\noindent
[24] Packwood D M 2012 Relaxation function for a continuous time random walk \textit{Hokkaido University Technical Report Series in Mathematics} \textbf{151} 168 - 171 \\

\noindent
[25] Durrett R 2010 \textit{Probability: Theory and Examples} (New York: Cambridge University Press)\\

\noindent
[26] Zwanzig R 2001 \textit{Nonequilibrium Statistical Mechanics} (New York: Oxford University Press)\\

\noindent
[27] Krojanski H G and Suter D 2004 Scaling of decoherence in wide NMR quantum registers \textit{Phys. Rev. Lett.} \textbf{93}, 90501 - 90505 \\

\noindent
[28] Stoneham M 2009 Is a room-temperature, solid-state quantum computer a mere fantasy? \textit{Physics} \textbf{2}, 34 - 40\\

\newpage
\noindent
\textbf{Figures} \\

\begin{figure}[h]
	\centering
		\includegraphics[width=1.00\textwidth]{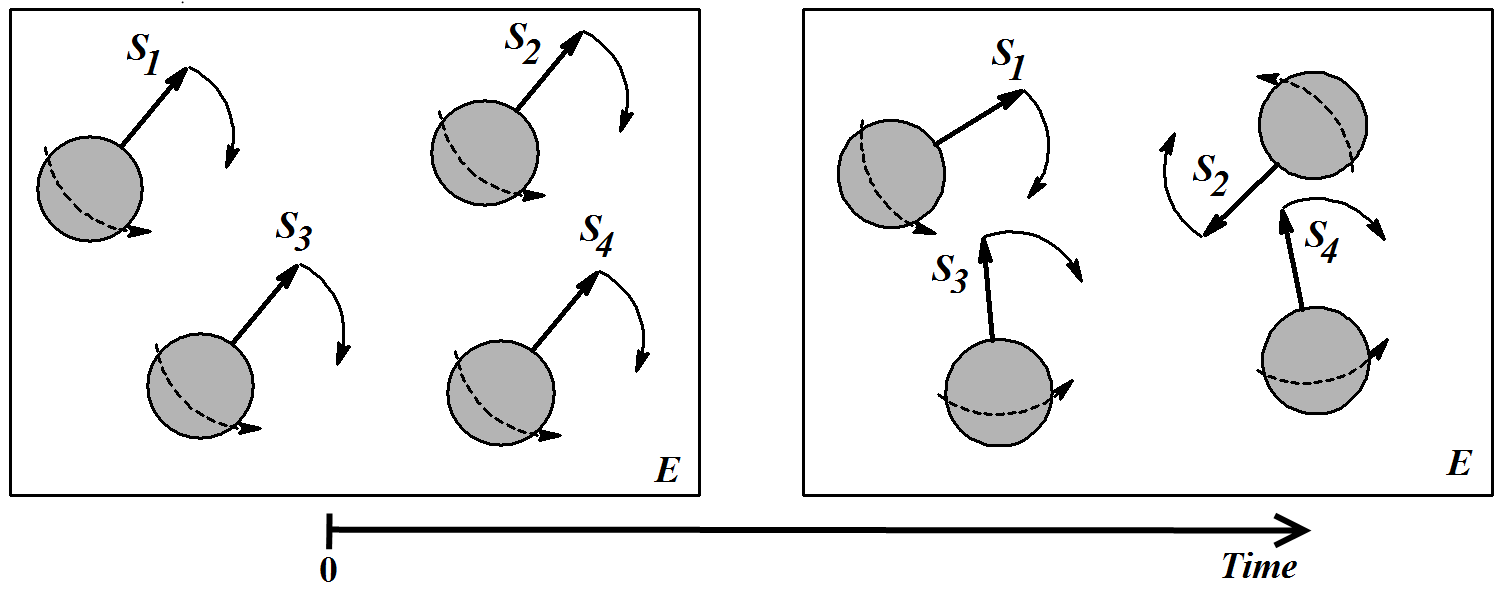}
	\caption{A simplified diagram of transverse relaxation in a magnetic resonance experiment. The diagrams show four spin particles contained in an environment $E$, drawn in the $xy$-plane. The spins are a long distance apart and so do not interact with each other, however they do interact with the degrees of freedom of $E$. The spin of the particles is represented by the particles rotating about their spin angular momentum vectors $S_i$ (this is indicated by the dotted arrows about each particle). At time zero an electromagnetic pulse is applied to the $xy$ plane. This places the spin angular momentum vectors in the $xy$-plane and aligns them in the same direction. The vectors rotate clockwise in plane, initially in unison (left hand diagram; this rotation represented by the solid curved arrows). However, random interactions between each spin and degrees of freedom of $E$ modulate the angular frequency of each rotation. These interactions differ from spin to spin, and so eventually the the rotating vectors fall out of alignment with each other and end up pointing in a variety of directions (right hand diagram).}
	\label{fig:Figure1}
\end{figure}

\begin{figure}[h]
	\centering
		\includegraphics[width=1.00\textwidth]{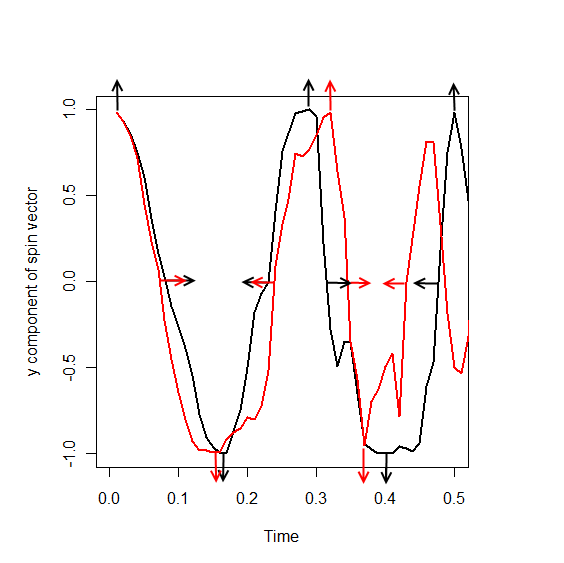}
	\caption{A simulation of $y$ component of the spin vector during a magnetic resonance experiment on two spins. One spin is represented by the black line and the other by the red line. The direction of the spin vector in the $xy$ plane at various points along the curves is indicated by the arrows. The curves are initially coincident, however they fall out of step with one another as random modulations of the rotational frequencies take place. The curves were simulated with the Kubo oscillator model with Wiener process frequency modulation and a mean frequency of $20$ (arbitrary units). See [16].}
	\label{fig:Figure2}
\end{figure}

\begin{figure}[h]
	\centering
		\includegraphics[width=1.00\textwidth]{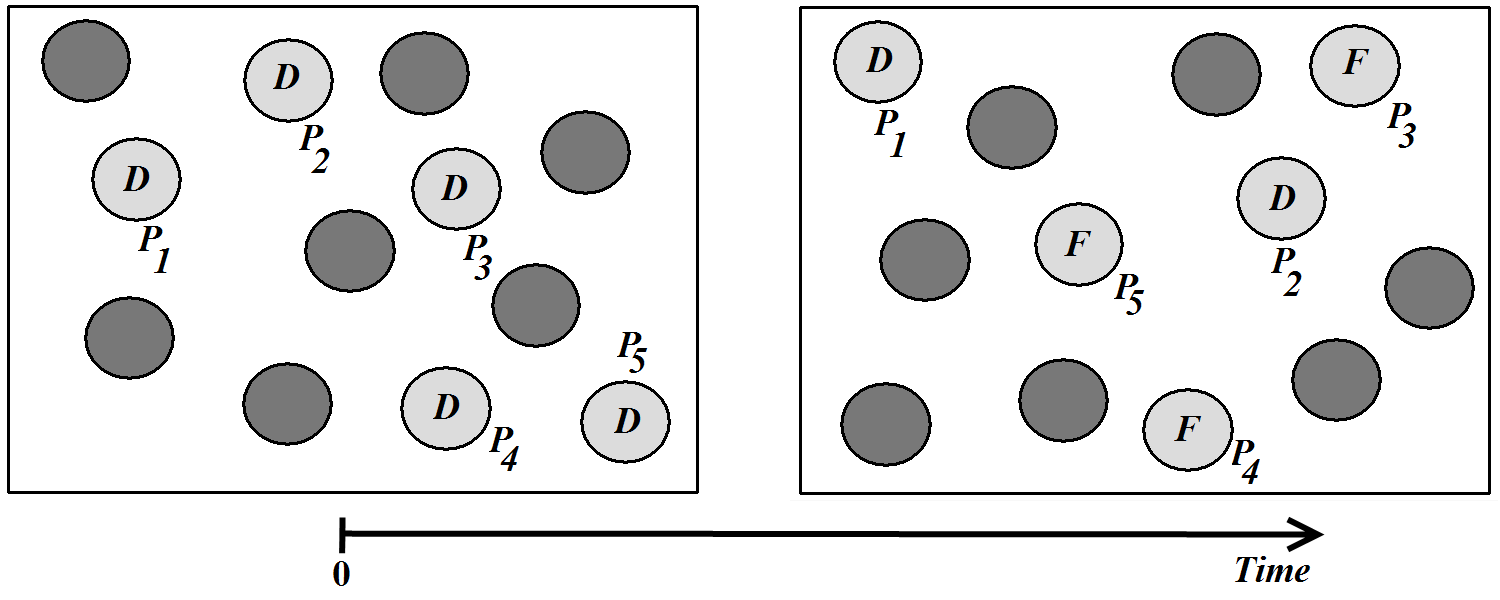}
	\caption{Thermal equilibration of a two dimensional two component gas. The dark grey particles represent the $Q$ component and the light grey particles represent the $P$ component. The kinetic energy space of the particles of $P$ is divided into two regions, $D$ and $F$. At time $0$ the particles of $P$ are instantaneously heated by an external source and are all placed in the $D$ state (left hand figure). However, random collisions with the particles of $Q$ cause the particles of $P$ to shift between the $D$ and $F$ states. Because the particles in $P$ each collide with the particles of $Q$ at different times, the particles in $P$ end up in a variety of different states at equilibrium (right hand diagram). The states $D$ and $F$ will be defined in Section 3 of the main paper.}
	\label{fig:Figure3}
\end{figure}

\begin{figure}[h]
	\centering
		\includegraphics[width=1.00\textwidth]{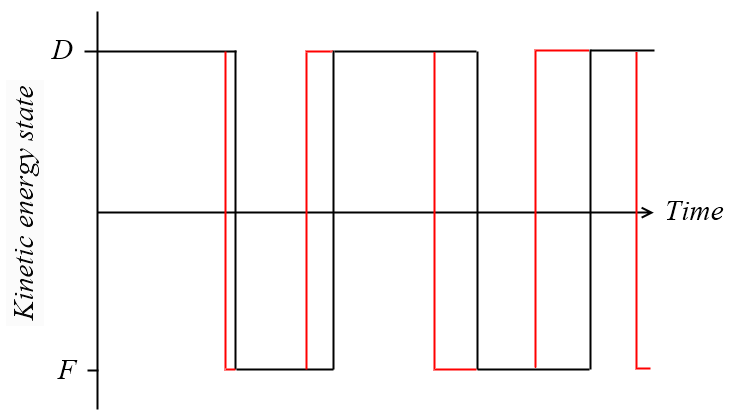}
	\caption{Sketch of the kinetic energy state of two molecules from $P$ as a function of time. The two molecules start out in the same kinetic energy state. However, because the two molecules collide with molecules from $P$ at different times, the oscillations between the two states falls out of step with each other over time. This switching between the $F$ and $D$ states is referred to as an `oscillation' in this work.}
	\label{fig:Figure4}
\end{figure}

\end{document}